\documentclass[11pt]{article}

\usepackage{graphicx,tikz,color,url}
\usepackage{amsmath,amssymb,latexsym,amsthm}
\usepackage{mathtools}
\usepackage{xspace}
\usepackage{subfig}

\usepackage{longtable}
\usepackage{booktabs}

\usepackage{color}
\usepackage{url}
\usepackage{algpseudocode}
\algnotext{EndFor}
\algnotext{EndIf}
\usepackage{tikz}
\usepackage{graphicx}
\usepackage{enumerate}
\usepackage{multirow}

\usepackage[figure,tworuled,linesnumbered,vlined]{algorithm2e}

\SetKwComment{tcp}{$\triangleright$ }{}%

\SetCommentSty{mycommfont}

\usepackage{hyperref}

\hypersetup{colorlinks=true, linkcolor=blue, citecolor=blue, urlcolor=blue}

\newtheorem{theorem}{Theorem}[section]
\newtheorem{definition}[theorem]{Definition}
\newtheorem{lemma}[theorem]{Lemma}
\newtheorem{proposition}[theorem]{Proposition}

\newtheorem{corollary}[theorem]{Corollary}
\newtheorem{remark}[theorem]{Remark}

\newcommand{\cp}{\,\square\,}

\newcommand{\diam}{{\rm diam}}

\newcommand{\adim}{{\rm adim}}
\newcommand{\Adim}{{\rm Adim}}

\newcommand{\BarabasiAlbert}{\mathit{BarabasiAlbert}}

\title{Burning some myths on privacy properties of social networks against active attacks}

\author{
Serafino Cicerone $^{a,}$\thanks{Email: \texttt{serafino.cicerone@univaq.it}}
\and
Gabriele {Di Stefano} $^{a,}$\thanks{Email: \texttt{gabriele.distefano@univaq.it}}
\and
Sandi Klav\v zar $^{b,c,d,}$\thanks{Email: \texttt{sandi.klavzar@fmf.uni-lj.si}}
\and
Ismael G. Yero $^{e,}$\thanks{Email: \texttt{ismael.gonzalez@uca.es}}
}

\begin{document}

\maketitle

\begin{center}
	$^a$ Department of Information Engineering, Computer Science, and Mathematics, \\
	     University of L'Aquila, Italy \\
	\medskip

	$^b$ Faculty of Mathematics and Physics, University of Ljubljana, Slovenia\\
	\medskip
	
	$^c$ Institute of Mathematics, Physics and Mechanics, Ljubljana, Slovenia\\
	\medskip
	
	$^d$ Faculty of Natural Sciences and Mathematics, University of Maribor, Slovenia\\
	\medskip

	$^e$ Departamento de Matem\'aticas, Universidad de C\'adiz, Algeciras Campus, Spain \\

\end{center}

\begin{abstract}
This work focuses on showing some arguments addressed to dismantle the extended idea about that social networks completely lacks of privacy properties. We consider the so-called active attacks to the privacy of social networks and the counterpart $(k,\ell)$-anonymity measure, which is used to quantify the privacy satisfied by a social network against active attacks. To this end, we make use of the graph theoretical concept of $k$-metric antidimensional graphs for which the case $k=1$ represents those graphs achieving the worst scenario in privacy whilst considering the $(k,\ell)$-anonymity measure. 

As a product of our investigation, we present a large number of computational results stating that social networks might not be as insecure as one often thinks. In particular, we develop a large number of experiments on random graphs which show that the number of $1$-metric antidimensional graphs is indeed ridiculously small with respect to the total number of graphs that can be considered. Moreover, we search on several real networks in order to check if they are $1$-metric antidimensional, and obtain that none of them are such. Along the way, we show some theoretical studies on the mathematical properties of the $k$-metric antidimensional graphs for any suitable $k\ge 1$. In addition, we also describe some operations on graphs that are $1$-metric antidimensional so that they get embedded into another larger graphs that are not such, in order to obscure their privacy properties against active attacks. 
\end{abstract}


\noindent
AMS Subj.\ Class.\ (2020):  05C12, 05C69, 05C76, 05C82, 05C85

\section{Introduction}

Nowadays, it is a widely accepted fact that social networks usually provide very low privacy features. The data collected in social media platforms is frequently used in ways that their users are unaware of. Third-parties (like advertisers, for instance) use them and, many times, do not care about who has access to them, and moreover, do not know whether they are also shared with malicious entities. Social networks are targets for cyberattacks, and so, some platforms are often compromised, which leads to some information, like private messages, photos, financial data, etc. could be exposed. One reason for this might be that they indeed lack of a clear transparency. That is, many social media analyzers often do not clearly explain or expose their practices with respect to data-sharing, or make this in a confusing manner that users avoid considering or even reading, and so, they give their consent to such companies to manage their data.

All these facts have created the ``myth'' that social networks present very low or even none privacy properties. However, there might be some arguments that could contribute to decrease these extended myths. In this work, we are precisely focused on presenting a few of these arguments, although we do not exactly state that social networks are secure with respect to the privacy they offer. Instead, we claim that they might be not that weak as it is usually understood. A reason for this is that we are considering only one kind of action among all that ones that malicious entities can perform in order to retrieve sensitive information.

In fact, there are many different styles of processing a social network so that hidden private information can be exposed. Entities dedicated to such malicious actions (frequently called attackers) might develop numerous attacks to a data set to retrieve some information from it. Among such actions, it is probably the most well-known that one called an active attack. Given a social network, an active attack to its privacy is (roughly speaking) an action that an entity can perform on such a graph to control a set of nodes of it, in order to detect or identify some elements of the graph. 
While the objectives of such an attacker may vary widely, the successful execution of malicious actions would inevitably compromise the privacy of users. Consequently, the existence of privacy-preserving methods and privacy measures for social networks is critically necessary.

The $k$-antiresolving sets and the $k$-metric antidimension of graphs were introduced in \cite{Trujillo-Rasua-2016} as the theoretical basis of the privacy measure $(k,\ell)$-anonymity for social networks under active attacks to their privacy. There was specifically stated that a given social network $G$ satisfies $(k,\ell)$-anonymity if $k$ is the smallest integer such that the $k$-metric antidimension of $G$ is at most $\ell$. This can be understood as follows. A given user of the social network $G$ has probability $1/k$ to be identified in such a social network under the assumption that there are $\ell$ attacker nodes in $G$. The number of attacker nodes $\ell$ in a social network is usually statistically assumed, since it is significantly smaller than the number of users of the network.

The privacy measure $(k,\ell)$-anonymity was indeed the first one of its type, and the first attempt into trying to quantify the privacy features that a social network achieves. Some improved variation of this measure was later on published in \cite{Mauw-2019}. Formal definitions of the terms above are as follows. Let $k\in\mathbb{N}$ and let $G=(V(G),E(G))$ be a connected graph.
\begin{itemize}
\item The \textit{metric representation} of a vertex $x\in V(G)$ with respect to an ordered set $S=\{v_1,\dots,v_t\}\subset V(G)$ is the vector $r(x|S)=(d_G(x,v_1),\dots,d_G(x,v_t))$, where $d_G(x,y)$ stands for the standard shortest-path distance between the two vertices $x,y$.
\item A set $S\subset V(G)$ is a $k$-\emph{antiresolving set} ($k$-\emph{ARS} for short) for $G$, if $k$ is the largest integer such that for all $u\notin S$ there exists a set $S_u\subseteq V\setminus (S\cup \{u\})$ with $|S_u|\geq k-1$ and where $r(u|S)=r(x|S)$ for every $x\in S_u$.
\item The $k$-\emph{metric antidimension} of $G$, denoted $\adim_k(G)$, is the cardinality of a smallest $k$-ARS for $G$.
\end{itemize}

The concepts above can be also understood in the following way. Given a vertex set $S\subset V(G)$ of a graph $G$, we define the following equivalence relation $\mathcal{R}_S$, where two vertices $x,y\in V(G)\setminus S$ are related by $\mathcal{R}_S$ if it follows that $r(x|S)=r(y|S)$. In this sense, given a set $S\subset V(G)$, we write $\mathcal{Z}_S=\{Z^1,\dots,Z^r\}$, for some $r\ge 1$, as the set of equivalence classes defined by $\mathcal{R}_S$. With such terminology in mind, we observe that an arbitrary set $S\subseteq V(G)$ is a $k$-ARS of $G$ with $k=\min\{|Z^i|\,:\,Z^i\in \mathcal{Z}_S\}$.

Having the concepts above in mind, it is known that a given social network $G$ achieves $(k,\ell)$-anonymity against active attacks to its privacy if $k$ is the smallest integer such that $\adim_k(G)\le \ell$. It might be then noticed that, in order to quantify the privacy achieved by a given social network $G$, it is necessary to know how to compute the $k$-metric antidimension of $G$. However, such a task might not be so efficiently made, since computing such a parameter for graphs is an NP-hard problem as independently proved in \cite{Chatterjee-2019,Zhang-2017}. In order to contribute to these computations, some bounds, approximations or heuristics are of interest. For instance, an ILP model was developed in \cite{fernandez-2023} that was used to compute the value of $\adim_k(G)$ for some random graphs, as well as, other (not polynomial) algorithms were implemented in \cite{DasGupta-2019} for similar computations. Bounds or closed formulas for this parameter are also known from \cite{fernandez-2023,Kratica-2019,Trujillo-Rasua-2016}.

It is natural to think that there is not a $k$-antiresolving set for every positive integer $k$ in a graph $G$. In this sense, by $\Adim(G)$ we represent the largest integer $k$ for which $G$ contains a $k$-ARS. We shall also say that a graph $G$ is {\em $k$-metric antidimensional} if $\Adim(G)=k$.

It was first noted in \cite{Trujillo-Rasua-2016} that any graph $G$ of maximum degree $\Delta$ satisfies that,
\begin{equation}
    \label{eq:bounds-Adim}
    1\le \Adim(G)\le \Delta(G).
\end{equation}
The equality for both bounds for $\Adim(G)$ occurs in several situations, and the case of the lower bound is of high interest based on the following argument. If $G$ is $1$-metric antidimensional, i.e., $\Adim(G)=1$, then this means that $G$ does not contain any $k$-ARS for every $k\ge 2$, or equivalently, $G$ only contains $1$-antiresolving sets. This is traduced into the fact that for any set of vertices $S\subset V(G)$, there is at least one vertex $v\notin S$ for which there are zero (0) vertices not in $S$ having the same metric representation as $v$, with respect to $S$. In other words, such a vertex is uniquely identified by the vertices in $S$, and so, if an attacker controls any set of vertices of such a graph $G$ (even of cardinality one), then the privacy of at least one vertex will be compromised.

In concordance with these facts, it seems that characterizing graphs $G$ that are $1$-metric antidimensional is worthwhile and even required, so that sensible data will not be made public throughout $1$-metric antidimensional social networks. Some first contributions in this direction were already given in \cite{trujillo-2016}, where all the trees and unicyclic graphs that are $1$-metric antidimensional were fully characterized and polynomial algorithms to detect them were designed.

On the other hand, based on the extended idea that social networks are rather weak with respect to their privacy, one might consider that the amount of such graphs ($1$-metric antidimensional) is very wide, and that, it should be also a reason that contributes to the extended myth regarding the wide weakness of social networks against active attacks to their privacy. The truth is that there is an infinite number of $1$-metric antidimensional graphs, which is already known for trees and unicyclic graphs from \cite{trujillo-2016}. However, as we show in our exposition, the quantity of them is indeed significantly low with respect to the total quantity of social networks that exists.

In addition, even so that there exist an infinite number of graphs that are $1$-metric antidimensional, one can always develop some operations for them, in order to be embedded into some larger graphs that are $k$-metric antidimensional for some $k\ge 2$. With such an embedding, we pretend to obscure the privacy features of $1$-metric antidimensional graphs into another larger graphs that are safer with respect to active attacks to their privacy. Along the way, in this work, we also describe several necessary and/or sufficient conditions for a graph to be $1$-metric antidimensional.

\section{Methodology}\label{sec:method}

The methods we use in our investigation are focused into three issues. First of all, in order to support the claims about breaking some myths on the privacy of social networks, we need a tool that will allow us to check whether a given social network is $1$-metric antidimensional or not. Fortunately, checking this fact is known to be polynomial, as shown in \cite{Chatterjee-2019} for a more general setting. This means that our methods are efficient enough so that a large number of computations can be performed. 

On a second hand, we are also interested into figuring out some structural properties of the graphs that $1$-metric antidimensional. To this end, we first consider some properties of graphs that are $k$-metric antidimensional for some $k\ge 2$, so that we further on focus on those graphs not satisfying these properties, and also separately, study mathematical properties of $1$-metric antidimensional graphs.

Finally, we center our attention into showing some operations that can be performed on a $1$-metric antidimensional graph so that it gets embedded into another graph that is not a $1$-metric antidimensional graph. This also contributes to our claim that social networks indeed do not lack too much of privacy properties. That is, for each $1$-metric antidimensional graph $G$ one can construct another graph $G'$ which is not a $1$-metric antidimensional graph, and such that $G$ is a subgraph of $G'$ and the properties of $G$ are obscured in the graph $G'$. In fact,  an infinite number of such supergraphs of $G$ can be constructed for each graph $1$-metric antidimensional graph $G$.

\subsection{The recognition algorithm}\label{sec:recognition}
Testing whether a given graph $G$ is $1$-metric antidimensional can be done by using the ADIM-1 algorithm provided in  Figure~\ref{alg:ADIM1}. This algorithm works as follows. At line~\ref{l:distances}, the all-pairs distances are computed. Then, the input graph is analyzed starting from each vertex $v$. In detail, the set $S$ is initially set as $S=\{v\}$. This set is intended as a potential $k$-ARS for $G$, for some $k\ge 1$. The while-loop at lines~\ref{l:while}--\ref{l:update-S} is responsible for computing for which $k$ the set $S$ is a $k$-ARS. If $k>1$ (see line~\ref{l:test-k}), $G$ is not $1$-metric antidimensional and hence the algorithm stops returning ``False'' (i.e., $\Adim(G) > 1$) and $k$. In such a case, $k$ is intended as a lower bound for the $k$-metric antidimensionality of $G$. Otherwise, when $k=1$, the algorithm adds all the sets in $\mathcal{Z}_S$ having cardinality $1$ to $S$ (see line~\ref{l:compute-Zi}), and repeats until $S=V(G)$. If the algorithm does not stop after checking all the vertices in $V(G)$, then it returns ``True'' (i.e., $\Adim(G) = 1$).

\begin{algorithm}[h]
\SetKwInput{Proc}{Algorithm}
\Proc{ADIM-1}
\SetKwInOut{Input}{Input}
\Input{An arbitrary graph $G$}
\SetKwInOut{Output}{Output}
\Output{True if and only if $\Adim(G)=1$}

\BlankLine
\BlankLine

    Compute the distance matrix of $G$ in $O\left(n^3\right)$ time using the Floyd-Warshall algorithm~\cite[p. 629]{CLR90} \label{l:distances}\;

    $S\leftarrow\emptyset$
    \tcp*[r]{we assume $S$ is a $k$-ARS for some $k\ge 1$} \label{l:init-S-empty}

    \ForEach{ $v_i\in V(G)$ \label{l:foreach-v} }{
        $S=\left\{ v_i \right\}$ \label{l:init-S-v}\;

        \While{ $\big(V\setminus S\neq\emptyset\big)$ \label{l:while} }{

            compute $k=\min\{|Z^i|\,:\,Z^i\in \mathcal{Z}_S\}$ \label{l:compute-k} \;

            \If{ $k>1$ \label{l:test-k} }{
                \KwRet (False, $k$)
                \tcp*[r]{$k$ is a lower bound for the $k$-metric antidimension of $G$}
             }
            \Else{
                let $Z^{i_1},\dots,Z^{i_\ell}$ be all equivalence classes in $\mathcal{Z}_S$ such that $\left|Z^{i_1}\right|=\dots=\left|Z^{i_\ell}\right|=1$ \label{l:compute-Zi}\;

                $S\,\leftarrow\,S\cup \left( \cup_{t=1}^\ell Z^{i_t} \right)$ \label{l:update-S} \;
            }
        }
    }
    \KwRet True \;
\caption{ Algorithm checking whether $\Adim(G)=1$.}
\label{alg:ADIM1}
\end{algorithm}

Algorithm ADIM-1 might indeed be seen as a special case of~\cite[Algorithm II]{Chatterjee-2019}, when looking for $1$-antidimensional graphs, and we refer the reader to this work for its correctness and computational time, which is $O(n^4)$.

\subsection{Other required terminology}

For an integer $n\ge 1$, we shall write $[n]$ to represent the set of integers $\{1,\dots,n\}$. Let $G$ be a connected graph. The \textit{order} of $G$ will be denoted by $n(G)$. For a given vertex $x\in V(G)$, its \textit{degree} is denoted by $\deg_G(x)$. The \textit{minimum} and \textit{maximum} \textit{degree} of $G$ are $\delta(G)$ and $\Delta(G)$, respectively. The \textit{distance} $d_G(u,v)$ between vertices $u$ and $v$ in $G$ is the usual shortest-path distance. The \emph{eccentricity} $\epsilon_G(u)$ of $u$ is the maximum distance between $u$ and any other vertex of $G$. The \emph{diameter} $\diam(G)$ of $G$ is the maximum of the eccentricities of the vertices of $G$. Also, the \textit{center} of $G$ is the set of vertices $x$ such that $\epsilon_G(x)=\min\{\epsilon_G(v)\,:\,v\in V(G)\}$.  If $x\in V(G)$, then $L_i(x)$ denotes the set of vertices of $G$ at distance $i$ from $x$. Also, as usual, $\kappa(G)$ denotes the (\textit{vertex}) \textit{connectivity} of a graph $G$. 
If $G$ is a graph that satisfies $\kappa(G)  \ge 2$, then $G$ is called a \textit{biconnected graph}. From \cite{trujillo-2016}, graphs $G$ satisfying that $\kappa(G)=1$ and $\Adim(G)=1$ are known (some trees and unicyclic graphs). However, not much more is known about $1$-metric antidimensional graphs $G$ with larger connectivity, i.e., with $\kappa(G)\ge 2$.

The {\em Cartesian product} $G \cp H$ has the vertex set $V(G) \times V(H)$, and vertices $(g,h)$ and $(g',h')$ are adjacent if $g = g'$ and $hh' \in E(H)$, or $h = h'$ and $gg' \in E(G)$. The {\em strong product} $G \boxtimes H$ is obtained from $G\cp H$ by adding to it the edges $(g,h)(g',h')$, where $gg' \in E(G)$ and $hh'\in E(H)$. The {\em lexicographic product} $G \circ H$ also has the vertex set $V(G) \times V(H)$, and vertices $(g,h)$ and $(g',h')$ are adjacent if $g = g'$ and $hh' \in E(H)$, or $gg' \in E(G)$.

\subsection{Plan of the exposition}

Once we have described our methodology in this section, we are then able to present the results of our investigation. In Section \ref{sec-exper}, we present several computational results that support our claim on burning the myth concerning the low privacy of social networks. In particular, we develop a large number of experiments on random graphs which show that the number of $1$-metric antidimensional graphs is indeed ridiculously small with respect to the total number of graphs that can be considered. In addition, we develop in Subsection \ref{subsec-real} some search on several real networks in order to check if they are $1$-metric antidimensional, and obtain that none of them are such. Further on, Sections \ref{sec-k-ARS} and \ref{sec:adim-1} contain some theoretical studies on the mathematical properties of the $k$-ARS of graphs for any suitable $k\ge 1$. Next, Section \ref{sec-increasing-k} is focused on developing some operations on graphs that are $1$-metric antidimensional so that they get embedded into another larger graphs that are not such, in order to obscure their privacy properties against active attacks. Such operations consist on constructing new graphs throughout making use of three of the four classical well--known product graphs, i.e., the lexicographic, the strong and the Cartesian products. We close our exposition with some concluding remarks and possible open problems that can be dealt with as a continuation to busting some more myths on the privacy of social networks.

\section{Experimental evaluation}\label{sec-exper}

In this section, we present the results of an extensive experimental evaluation, whose main objective is to explore the class of $1$-metric antidimensional graphs and to explode the widely extended myth on the weakness of social networks concerning their privacy. Initially, we performed an exhaustive search of these graphs among all graphs with order at most 11. Then, we tested both graphs generated according to well-known random models (i.e., Barab\'{a}si-Albert, Erd\"os-R\'{e}nyi, Erd\"os-R\'{e}nyi-Gilbert) and real-world graphs taken from publicly available repositories. All the experiments have been carried out using the ADIM-1 algorithm described in Subsection~\ref{sec:recognition}. This algorithm and all the testing scripts have been implemented under SageMath, the Sage Mathematics Software System (version 10.5).

\subsection{Exhaustive search}
Table~\ref{tab:exhaustive_search} presents the results obtained by performing an exhaustive search on all graphs $G$ with order between 3 and 11.

\begin{table}[h]
\footnotesize
    \centering
    \begin{tabular}{lrrrrrc}
        \toprule
        $n$ & \textit{total distinct} & \textit{connected} & \textit{found} & \textit{ratio} & \textit{connectivity}  & \textit{max density}\\
        \midrule
        3 & 4 & 2 & 0 & -- & -- & --\\
        4 & 11 & 6 & 1 & 0.166666 & 1: 1  & 0.50 \\
        5 & 34 & 21 & 0 & -- & -- &  -- \\
        6 & 156 & 112 & 1 & 0.008928 & 1: 1  & 0.33 \\
        7 & 1,044 & 853 & 2 & 0.002344 & 1: 2  & 0.33 \\
        8 & 12,346 & 11,117 & 13 & 0.001169 & 1: 13  & 0.35 \\
        9 & 274,668 & 261,080 & 110 & 0.000421 & 1: 110  & 0.38\\
        10 & 12,005,168 & 11,716,571 & 1,894 & 0.000161 & 1: 1,884; 2: 10  & 0.40\\
        11 & 1,018,997,864 & 1,006,700,565 & 52,842 & 0.000052 & 1: 52,505; 2: 337  & 0.42 \\
        \bottomrule
    \end{tabular}
    \caption{Exhaustive search for graphs with order in the range $[3,11]$}
    \label{tab:exhaustive_search}
\end{table}

The first column indicates the order of the graphs, the second column displays the total number of \emph{distinct} (non-isomorphic) graphs for that order, while the third column shows the count of connected graphs among these distinct ones. The subsequent columns detail the following findings: \emph{found} represents the number of 1-metric antidimensional graphs identified, the \emph{ratio} column gives the proportion of \emph{found} graphs to \emph{connected} graphs, and the \emph{connectivity} column specifies the number of found graphs with a given connectivity. Notably, there are no biconnected 1-metric antidimensional graphs for orders less than 10. For order 10, 10 biconnected graphs were found out of 1894 graphs, and for order 11, 237 biconnected graphs were found out of 52842 graphs (these counts are included in the \emph{found} column and further specified in the \emph{connectivity} column). In the performed exhaustive search, no 1-metric antidimensional graphs that are 3-connected were found. The last column shows the \emph{max-density} among the found graphs; given a graph $G$, by density, we mean the ratio between the number of edges of $G$ and maximum number of possible edges, i.e., $n(G)(n(G)-1)/2$.

Notice that there are no graphs of order 3 and 5 that are 1-metric antidimensional. The graphs of order 4 and 6 are the paths $P_4$ and $P_6$, respectively. The two graphs having order 7 are both formed by $P_6$ plus an additional vertex which is adjacent either to the third vertex of the path, or to the second and third vertices of the path. The 10 found graphs of order 10 that are both $1$-metric antidimensional and biconnected are shown in Figure~\ref{fig:bi-connected}.

\begin{figure}[h!]
    \centering
\includegraphics[width=0.60\linewidth]{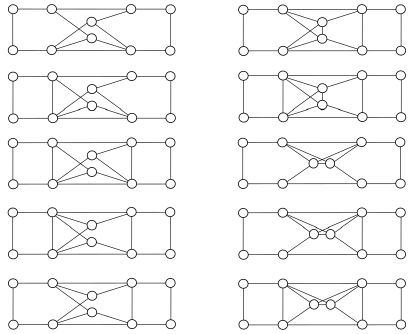}
    \caption{\small The ten biconnected graphs $G$ having $\Adim(G)=1$.}
    \label{fig:bi-connected}
\end{figure}

Table~\ref{tab:exhaustive_search-density} illustrates, for each order $n$ from 6 to 10, the densities at which 1-metric antidimensional graphs were identified. Specifically, for each such density, the table reports the count of 1-metric antidimensional graphs found and the count of all other graphs (i.e., those that are $k$-metric antidimensional, where $k>1$).

\begin{table}[h!]
\footnotesize
    \centering
    \begin{tabular}{ccrr}
        \toprule
        $n$ & \textit{density} & \textit{found} & \textit{others}   \\
        \toprule
        6 &  0.33  & 1  &  5 \\
        \midrule
        7 &  0.28  & 1  & 10 \\
          &  0.33  & 1  & 32 \\
        \midrule
        8 &  0.25  & 2  & 21 \\
          &  0.28  & 3  & 86 \\
          &  0.32  & 6  & 230 \\
          &  0.25  & 2  & 484 \\
        \midrule
        9 &  0.22  & 2  & 45 \\
          &  0.25  & 4  & 236 \\
          &  0.37  & 11 & 786 \\
          &  0.30  & 25 & 2,050 \\
          &  0.33  & 37 & 4,458 \\
          &  0.36  & 24 & 8,380 \\
          &  0.38  & 7  & 13,848 \\
        \midrule
       10 &  0.20  & 6   & 100 \\
          &  0.22  & 14  & 643 \\
          &  0.24  & 24  & 2,654 \\
          &  0.26  & 56  & 8,492 \\
          &  0.28  & 146 & 22,804 \\
          &  0.31  & 298 & 53,565 \\
          &  0.33  & 462 & 112,156 \\
          &  0.35  & 459 & 211,407 \\
          &  0.37  & 288 & 361,054 \\
          &  0.40  & 117 & 560,989 \\
          &  0.42  & 24  & 795,606 \\
        \bottomrule
    \end{tabular}
    \caption{\small Details about some results of the exhaustive search.}
\label{tab:exhaustive_search-density}
\end{table}

\subsection{Testing randomly generated graphs}
As already said, we tested graphs generated according to the well-known Barab\'{a}si-Albert and Erd\"os-R\'{e}nyi random models.

\paragraph{Barab\'{a}si-Albert model.}
We started by using the Barab\'{a}si–Albert model~\cite{BarabasiAlbert-2002}. This model allows the generation of random scale-free networks using a preferential attachment mechanism. Preferential attachment means that the more connected a node is, the more likely it is to receive new links. In particular, $\BarabasiAlbert(n,m)$ mechanism creates a graph with $n$ vertices incrementally: at each step, add one new node, then sample $m$ neighbors among the existing vertices from the network, with a probability that is proportional to the number of links that the existing nodes already have.

We tested graphs with the order in the range [11,100]. For each fixed order $n$, we set $m=2$ and generated 2 millions random graphs. Note that $m=1$ generates only trees, and $m=2$ is recommended to model social networks. Among the 180 millions of generated graphs, we found only 83 1-metric antidimensional graphs (of which, only 68 were distinct). Concerning size and connectivity, details about the found graphs are reported in Table~\ref{tab:barabasi}.

\begin{table}[h!]
\footnotesize
    \centering
    \begin{tabular}{crl}
        \toprule
        $n$ & \textit{found} & \textit{connectivity}  \\
        \midrule
        11 & 18 & 1: 11; 2: 7  \\
        12 & 19 & 1: 8; 2: 11  \\
        13 & 21 & 1: 13; 2: 8  \\
        14 & 9  & 1: 4; 2: 5  \\
        15 & 1  & 1: 1  \\
        \bottomrule
    \end{tabular}
    \caption{\small Details about the 1-metric antidimensional graphs found with the Barab\'{a}si–Albert model.}
    \label{tab:barabasi}
\end{table}

\paragraph{Erd\"os-R\'{e}nyi model.}
We refer to the well-known model introduced by Paul Erd\"os and Alfr\'{e}d R\'{e}nyi in 1959 for generating random graphs~\cite{ErdosRenyi-1959}. In this model, all graphs on a fixed vertex set with a fixed number of edges are equally likely. In particular, in the $G(n,m)$ model, a graph is chosen uniformly at random from the collection of all graphs which have $n$ nodes and $m$ edges.

Regarding the experiments conducted with this model, the substantial required execution time led us to initially generate graphs with an order n in the small range of $[11,16]$. For each fixed value of $n$, we then varied the number of edges $m$ across a wide range, specifically from $2n$ up to $n(n-1)/2$. Subsequently, for every pair of $(n,m)$ values, we generated 1,000,000 random instances.  Results about this preliminary test are reported in Table~\ref{tab:Gnm}.

\begin{table}[h!]
\footnotesize
        \centering
        \begin{tabular}[t]{rr}
            \toprule
            \textit{property} & \textit{value}  \\
            \midrule
            generated & 190,000,000  \\
            connected & 189,120,028   \\
            distinct found & 3,765   \\
            \bottomrule
        \end{tabular}
\quad\quad
        \begin{tabular}[t]{crl}
            \toprule
            $n$ & \textit{found} & \textit{connectivity}  \\
            \midrule
            11 & 198   & 1: 193, 2: 5  \\
            12 & 675   & 1: 668, 2: 7  \\
            13 & 1,006 & 1: 996, 2: 10   \\
            14 & 901   & 1: 892, 2: 9   \\
            15 & 598   & 1: 595, 2: 3   \\
            16 & 387   & 1: 387   \\
            \bottomrule
        \end{tabular}
        \caption{\small Results about graphs generated with the $G(n,m)$ model.}
        \label{tab:Gnm}
\end{table}

Next, we repeated the same type of experiments but for values of $n$ equal to 20, 30, 40 and 50. For such values, the number of found graphs (along with connectivity) is distributed as follows: 20: (1: 254), 30: (1: 444), 40: (1: 199), and 50: (1: 166). It is worth noticing that starting from $n=16$, only graphs with connectivity 1 were found. As a last observation, for each order $n$ tested, Table~\ref{tab:Gnm-density} presents the density of the found 1-metric antidimensional graph having the maximum number of edges.

\begin{table}[h!]
\footnotesize
    \centering
    \begin{tabular}{cccc}
        \toprule
        $n$ & \textit{largest} $m$ &  $n(n-1)/2$ & \textit{density}   \\
        \midrule
        11 & 25  & 55  & 0.45\\
        12 & 31  & 66  & 0.47\\
        13 & 34  & 78  & 0.43\\
        14 & 40  & 91  & 0.44\\
        15 & 45  & 105  & 0.43\\
        16 & 48  & 120  & 0.40\\
        20 & 82  & 190  & 0.43\\
        30 & 95  & 435  & 0.22\\
        40 & 133 & 780  & 0.17\\
        50 & 163 & 1225 & 0.13\\
        \bottomrule
    \end{tabular}
    \caption{\small Additional results about the $G(n,m)$ model. For each $n$, the density of the found 1-metric antidimensional graph having the maximum number of edges is reported.}
    \label{tab:Gnm-density}
\end{table}

\paragraph{Erd\"os-R\'{e}nyi-Gilbert model.}
We refer to the well-known model introduced by Gilbert for generating random graphs. It is also called the Erd\"os–R\'{e}nyi–Gilbert model (see~\cite{Fienberg-2012}) or the $G(n,p)$ model. According to such notation, a graph of order $n$ is constructed by connecting nodes randomly. Each edge is included in the graph with probability $p$, independently from every other edge. Using the $G(n,p)$ model, we performed two kinds of experiments. 

In the first type of experiment, we generated graphs with the order $n$ ranging from 11 to 100. While varying $n$, we initially fixed the probability $p$ at $0.25$. Then, for each order $n$, we generated 2 millions random graphs. After this first phase, we repeated the same approach three times, varying the probability within the set of values $\{0.20, 0.15, 0.10\}$. In summary, for each probability value, we generated a total of 180 millions of random graphs. General results about these experiments are reported in Table~\ref{tab:Gnp-general}, while more detailed findings can be found in Table~\ref{tab:Gnp-details}. 

\begin{table}[h]
\footnotesize
        \centering
        \begin{tabular}{rrrrr}
            \toprule
            \textit{property} & $p=0.25$ & $p=0.20$ & $p=0.15$ & $p=0.10$ \\
            \midrule
            generated 		& 180,000,000 & 180,000,000  & 180,000,000 & 180,000,000 \\
            connected 		& 174,497,428 & 167,659,141  & 152,927,782 & 118,197,366 \\
            total found 		& 32,328      & 58,121       & 59,666      &      23,900 \\
            distinct found 	& 19,077      & 35,286       & 41,866      &      21,004 \\
            \bottomrule
        \end{tabular}
		\quad\quad
        \caption{\small General results about graphs generated with the $G(n,p)$ model.}
        \label{tab:Gnp-general}
\end{table}

\begin{table}[h!]
\footnotesize
\centering
\begin{tabular}{c|rl|rl|rl|rl}
\toprule
    & \multicolumn{2}{c}{$p=0.25$} & \multicolumn{2}{c}{$p=0.20$} & \multicolumn{2}{c}{$p=0.15$} & \multicolumn{2}{c}{$p=0.10$}\\ \cline{2-9} 
$n$ & \textit{found} & \textit{connectivity} & \textit{found} & \textit{connectivity} & \textit{found} & \textit{connectivity} & \textit{found} & \textit{connectivity}\\
\midrule
 11 &  2,571   & 1: 2,494; 2: 7 	& 1378 & 1: 1378 		& 546  & 1: 546	 		&  157 & 1: 157	 	\\
 12 &  4,158   & 1: 4,157, 2: 1 	& 3502 & 1: 3501, 2: 1	& 1741 & 1: 1741		&  359 & 1: 359		\\
 13 &  4,917   & 1: 4,914, 2: 3  & 6285 & 1: 6284, 2: 1	& 3645 & 1: 3645		&  612 & 1: 612		\\
 14 &  3,279   & 1: 3,277, 2: 2 	& 6942 & 1: 6942		& 5155 & 1: 5155		&  849 & 1: 849		\\
 15 &  1,967   & 1: 1,967 		& 5656 & 1: 5656		& 5621 & 1: 5620, 2: 1 	& 1095 & 1: 1095	\\
 16 &  1,123   & 1: 1,123  		& 4103 & 1: 4103		& 5328 & 1: 5328		& 1207 & 1: 1207	\\
 17 &  534     & 1: 534  		& 2720 & 1: 2719, 2: 1 	& 4584 & 1: 4584		& 1289 & 1: 1289    \\
 18 &  308     & 1: 308  		& 1741 & 1: 1741		& 3843 & 1: 3843		& 1400 & 1: 1400    \\
 19 &  150     & 1: 150  		& 1171 & 1: 1171 		& 3023 & 1: 3023 		& 1393 & 1: 1393    \\
 20 &  73      & 1: 73  			&  730 & 1: 730 		& 2318 & 1: 2318 		& 1345 & 1: 1345 	\\
 21 &  38      & 1: 38  			&  413 & 1: 413 		& 1709 & 1: 1709 		& 1327 & 1: 1327 	\\
 22 &  13      & 1: 13  			&  247 & 1: 247 		& 1264 & 1: 1264 		& 1331 & 1: 1331 	\\
 23 &   8      & 1: 8  			&  154 & 1: 154 		&  908 & 1: 908 		& 1167 & 1: 1167 	\\
 24 &   4      & 1: 4  			&  100 & 1: 100 		&  652 & 1: 652 		& 1060 & 1: 1060 	\\
 25 &   3      & 1: 3  			&   59 & 1: 59 			&  470 & 1: 470 		&  988 & 1: 988 	\\
 26 &   1      & 1: 1 			&   32 & 1: 32 			&  324 & 1: 324 		&  886 & 1: 886 	\\
 27 &   0      & -- 				&   25 & 1: 25 			&  231 & 1: 231 		&  765 & 1: 765 	\\
 28 &   0      & -- 				&   11 & 1: 11			&  165 & 1: 165			&  668 & 1: 668		\\ 
 29 &   0      & -- 				&    5 & 1: 5 			&  100 & 1: 100 		&  554 & 1: 554 	\\
 30 &   0      & -- 				&    7 & 1: 7			&   67 & 1: 67			&  501 & 1: 501		\\
 31 &   0      & -- 				&    3 & 1: 3 			&   50 & 1: 50 			&  360 & 1: 360 	\\
 32 &   0      & -- 				&    0 & --			    &   27 & 1: 27			&  325 & 1: 325		\\                        
 33 &   0      & -- 				&    1 & 1: 1			&   34 & 1: 34			&  248 & 1: 248		\\
 34 &   0      & -- 				&    0 & --			    &   18 & 1: 18			&  199 & 1: 199		\\
 35 &   0      & -- 				&    0 & --			    &    8 & 1: 8			&  169 & 1: 169		\\
 36 &   0      & -- 				&    0 & --			    &   11 & 1: 11			&  145 & 1: 145		\\
 37 &   0      & -- 				&    0 & --			    &    8 & 1: 8			&  114 & 1: 114		\\
 38 &   0      & -- 				&    0 & --			    &    7 & 1: 7			&   94 & 1: 94		\\                   
 39 &   0      & -- 				&    1 & 1: 1 			&    2 & 1: 2 			&   87 & 1: 87 		\\            
 40 &   0      & -- 				&    0 & --			    &    3 & 1: 3			&   66 & 1: 66		\\                   
 41 &   0      & -- 				&    0 & --			    &    1 & 1: 1			&   49 & 1: 49		\\                   
 42 &   0      & -- 				&    0 & --			    &    1 & 1: 1			&   43 & 1: 43		\\                   
 43 &   0      & -- 				&    0 & --			    &    0 & --				&   38 & 1: 38		\\                   
 44 &   0      & -- 				&    0 & --			    &    0 & --				&   33 & 1: 33		\\                   
 45 &   0      & -- 				&    0 & --			    &    1 & 1: 1			&   18 & 1: 18		\\  
 46 &   0      & -- 				&    0 & --			    &    0 & --				&   15 & 1: 15		\\  
 47 &   0      & -- 				&    0 & --			    &    0 & --				&   19 & 1: 19		\\  
 48 &   0      & -- 				&    0 & --			    &    0 & --				&    6 & 1: 6		\\  
 49 &   0      & -- 				&    0 & --			    &    0 & --				&   11 & 1: 11		\\                 
 50 &   0      & -- 				&    0 & --			    &    1 & 1: 1			&    3 & 1: 3		\\                                            
 51 &   0      & -- 				&    0 & --			    &    0 & --				&    2 & 1: 2		\\                                            
 52 &   0      & -- 				&    0 & --			    &    0 & --				&    2 & 1: 2		\\                                            
 54 &   0      & -- 				&    0 & --			    &    0 & --				&    3 & 1: 3		\\                                            
 55 &   0      & -- 				&    0 & --			    &    0 & --				&    2 & 1: 2		\\                                             
  \bottomrule
\end{tabular}
\caption{\small Detailed results about graphs generated with the $G(n,p)$ model.}
\label{tab:Gnp-details}
\end{table}

Concerning Table~\ref{tab:Gnp-details}, it is interesting to note that for $n>55$, no 1-metric antidimensional graphs were found regardless of the used probability. Moreover, the connectivity is still limited to 2, with a minimal number of cases found up to $n=17$. Due to the limited number of graphs found as $n$ increases, we repeated the same experiment for the specific value of $n=30$ and $p=0.25$, extending to 50 millions the number of generated graphs. With this kind of input, we found only one graph that is 1-metric antidimensional (and the connectivity of this graph is still 1). 

In the second kind of experiments, we varied $n$ in the range $[11,50]$, but this time we used a specific value of $p$ for each value of $n$. As it is established that a graph in $G(n,p)$ is almost surely connected when $p > \frac{ (1+\varepsilon )\ln n }{n}$, we choose $p = \frac{ 1.001\ln n }{n}$. The number of graphs found for each value of $n$ resulting from this experiment are reported as charts in Figure~\ref{fig:Gnp-varying-p}.

\begin{figure}[h]
\centering
\includegraphics[width=0.80\textwidth]{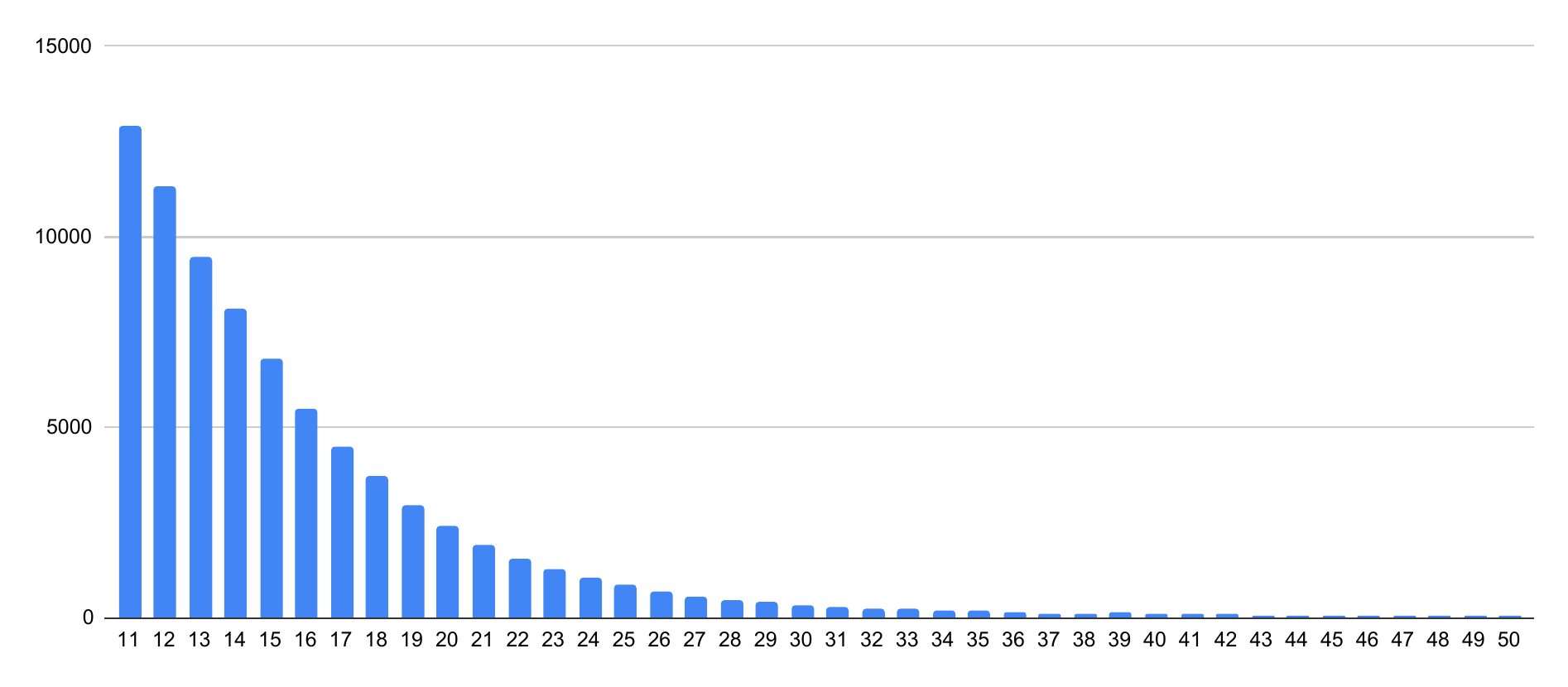}
\includegraphics[width=0.80\textwidth]{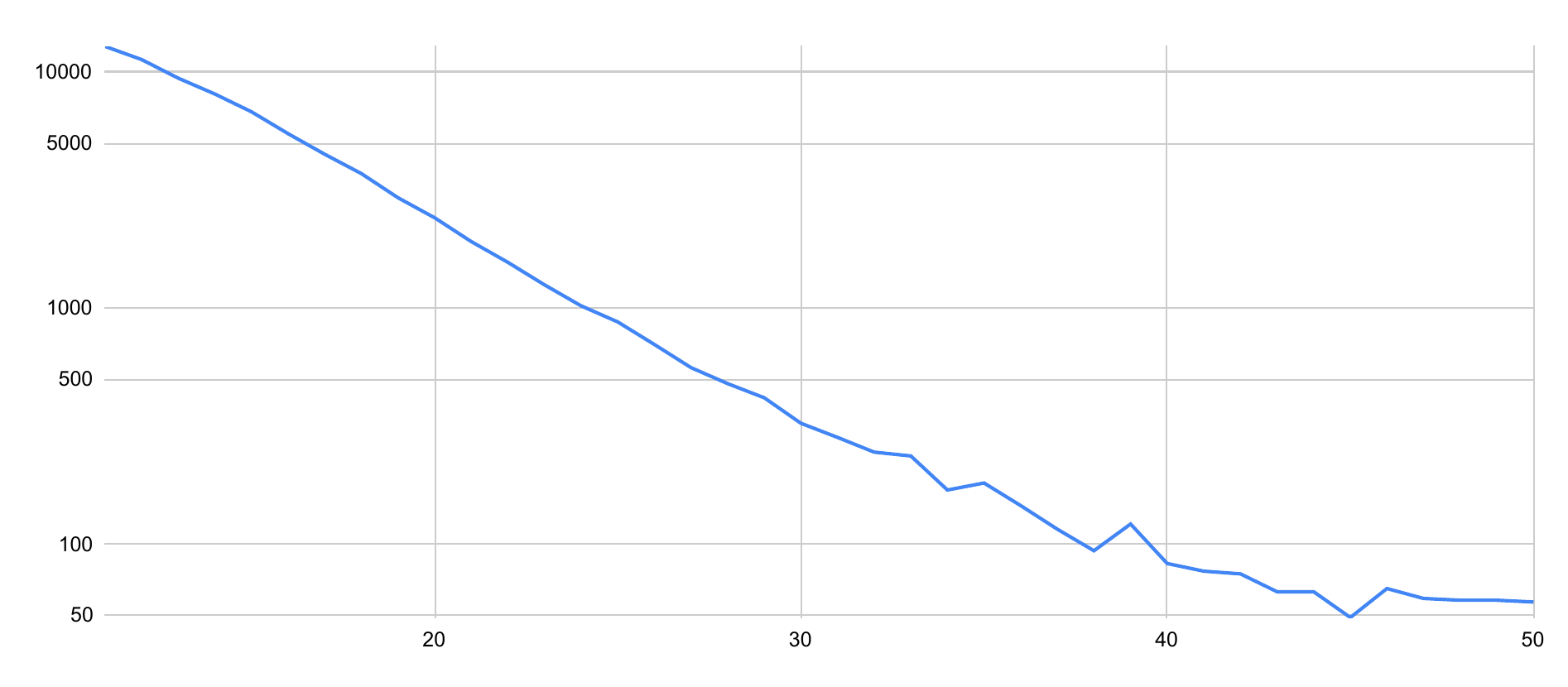}
\caption{\small Results about the second kind of experiments with the $G(n,p)$ model, where $n\in [11,50]$ and $p = \frac{ 1.001\ln n }{n}$. Notice that the second chart reports the same data but as a line chart with log scale on the vertical axis.}
\label{fig:Gnp-varying-p}
\end{figure}

\subsection{Testing real world networks}\label{subsec-real}
We report some results of tests performed on real networks, using some publicly available datasets.
The first dataset is named \emph{Gemsec Facebook}~\cite{RozemberczkiEtAl-2019}. This dataset, obtained from the Stanford Network Analysis Project (SNAP) at \url{https://snap.stanford.edu/}, comprises data collected from Facebook pages. Specifically, it includes 8 graphs, each representing the network of verified Facebook pages within a different category. In these graphs, nodes represent the pages, and edges indicate mutual likes between them. To enhance anonymity, the nodes have been reindexed.

The second dataset refers to some Ego networks from Facebook, that is, networks given by the connections among nodes sharing a common neighbor (the ego node). They are also provided by SNAP (\url{https://snap.stanford.edu/data/ego-Facebook.html}). The network ``facebook combined" combines the data of ten Ego networks.

The last three datasets came from the Network Repository Project, an initiative that
collects hundreds of network datasets (\url{https://networkrepository.com/}). The first of them refers to Facebook networks and the second to email networks. As a further dataset, we decided to test the 1-metric antidimensionality for non-social networks:  we chose some power networks from \url{https://networkrepository.com/power.php}.

Table~\ref{tab:real-networks} presents details about some networks of the used datasets.  Notice that most of them are 1-connected graphs. Each graph in the table has been tested to verify if it is 1-metric antidimensional. As a final result, none of them was found to be 1-metric antidimensional.

\begin{table}[t]
\footnotesize
    \centering
    \begin{tabular}{lcrrrrc}
\toprule
\textit{category} & \textit{network} &  $n$  & $m$ & \textit{density} & $\Delta$ & $\delta$ \\
\midrule
Gemsec FB~\cite{RozemberczkiEtAl-2019} & government & 7057 & 89455 & 0.00359 & 697 & 1  \\
& athletes & 13866 & 86858 & 0.00090  & 468 & 1  \\
& politician  & 5908 & 41729 & 0.00239 & 323 & 1  \\
& public figure  & 11565 & 67114 & 0.00100 & 326 & 1  \\
& company  & 14113 & 52310 & 0.00053 & 215 & 1  \\
& artist  & 50515 & 819306 & 0.00064 & 1469 & 1  \\
& tvshow  & 3892 & 17262 & 0.00228 & 126 & 1  \\
& new sites  & 27917 & 206259 & 0.00053 & 678 & 1  \\
\midrule
Ego FB~\cite{McAuley-2012} & facebook combined & 4039 & 88234 & 0.01082 & 1045 & 1  \\
& 107 & 1034 & 26749 & 0.05009 & 253 & 1\\
& 348 & 224 & 3192 & 0.12780 & 99 & 1\\
& 686 & 168 & 1656 & 0.11805 & 77 & 1\\
\midrule
Social FB~\cite{Rossi-2015} &socfb-Auburn71 & 18448 & 973918 & 0.00572 & 5160 & 1  \\
&socfb-CMU & 6621 & 249959 & 0.01141 & 840 & 1  \\
&socfb-Amherst41 & 2235 & 90954 & 0.03643 & 467 & 1  \\
\midrule
Email~\cite{Rossi-2015} & email-enron-only & 143 & 623 & 0.06136 & 42 & 1  \\
& email-univ & 1133 & 5451 & 0.00850 & 71 & 1  \\
& email-enron-large & 33696 & 180811 & 0.00032 & 1383 & 1  \\
& email-EU & 32430 & 54397 & 0.00010 &  623 & 1 \\
\midrule
Power Network~\cite{Rossi-2015} & power-US-Grid & 4941 & 6594 & 0.00054 & 19 & 1  \\
& power-bcspwr10 & 5300 & 13571 & 0.00100 & 15 & 3  \\
& power-bcspwr09 & 1723 & 4117 & 0.00278 & 16 & 3  \\
\bottomrule
\end{tabular}
    \caption{Data about real network datasets}
    \label{tab:real-networks}
\end{table}

\section{Mathematical properties of $k$-ARS for larger values of $k$}\label{sec-k-ARS}

In order to find some mathematical properties of $k$-ARS in a graph, in this section, we connect the existence of $k$-ARS with some classical areas of graph theory, including vertex connectivity and modular decompositions. In addition, the results of this section shall be further used, while we deal with increasing the privacy properties of graphs that are $1$-metric antidimensional.

\subsection{Vertex connectivity}

As we next show, the (vertex) connectivity of graphs can be used to bound the value $\Adim(G)$ for a given graph $G$. The next concept can also be used for a similar purpose. If $x$ is a vertex of a connected graph $G$, then we set
$$\#\epsilon_G(x) = |\{y:\ d_G(x,y) = \epsilon_G(x)\}|\,$$
and
$$\#\epsilon(G) = \max_{x\in V(G)} \#\epsilon_G(x)\,.$$
Note that if $\diam(G) = 2$, then $\#\epsilon_G(x) = n(G) - \deg_G(x) - 1$ which in turn implies that
\begin{equation}
\label{eq:sharp-ecc-diam-2}
\#\epsilon(G) = n(G) - \delta(G) - 1\,.
\end{equation}

\begin{lemma}
\label{lem:k-connected}
If $G$ is a connected graph, then
$$\Adim(G) \ge \min \{\kappa(G), \#\epsilon(G)\}\,.$$
Moreover, if $k = \min \{\kappa(G), \#\epsilon(G)\}$, then $\adim_k(G) = 1$.
\end{lemma}

\begin{proof}
Let $x\in V(G)$ be a vertex with $\#\epsilon_G(x) = \#\epsilon(G)$. Since each $L_i(x)$, $i\in [\epsilon_G(x)-1]$, is a cut set, we have $|L_i(x)| \ge \kappa(G)$ for every $i\in [\epsilon_G(x)-1]$. Since $\#\epsilon_G(x) = \#\epsilon(G)$ we also have $|L_{\epsilon_G(x)}(x)| = \#\epsilon(G)$. It follows that $\{x\}$ is a $\min \{\kappa(G), \#\epsilon(G)\}$-ARS which implies both assertions of the lemma.
\end{proof}

Lemma~\ref{lem:k-connected} implies that if $G$ is a graph with $\Adim(G) = 1$ and $\kappa(G)\ge 2$, then for every vertex $x$ there exists a unique vertex $y$ such that $d_G(x,y) = \epsilon_G(x)$. That is, the following holds.

\begin{corollary}
\label{cor:ecc-to-be-1}
If $G$ is a graph with $\Adim(G) = 1$ and $\kappa(G)\ge 2$, then $\#\epsilon(G) = 1$.
\end{corollary}

The next consequence of Lemma~\ref{lem:k-connected} also deserves special attention.

\begin{corollary}
\label{cor:kappa-regular-diam-2}
If $G$ is a $\kappa(G)$-regular graph with $\diam(G) = 2$ and $\kappa(G) \le  \frac{n(G)-1}{2}$, then $\Adim(G) = \kappa(G)$.
\end{corollary}

\begin{proof}
Combining Lemma~\ref{lem:k-connected}, Eq.~\eqref{eq:sharp-ecc-diam-2}, and the assumption $\kappa(G) \le \frac{n(G)-1}{2}$, we get
\begin{align*}
\Adim(G)
& \ge \min \{\kappa(G), \#\epsilon(G)\} \\
& =
\min \{\kappa(G), n(G) - \delta(G) - 1\} \\
& = \min \{\kappa(G), n(G) - \kappa(G) - 1\} \\
& = \kappa(G)\,.
\end{align*}
On the other hand, by \eqref{eq:bounds-Adim} we have $\Adim(G) \le \Delta(G) = \kappa(G)$, and we are done.
\end{proof}

For example, if $P$ denotes the Petersen graph $P$, then Corollary~\ref{cor:kappa-regular-diam-2} yields $\Adim(P) = 3$.

\subsection{Modular decompositions}

The modular decomposition is a decomposition of a graph into subsets of vertices called modules. Given a graph $G$, a set $M\subseteq V(G)$ is a {\em module} of $G$ if the vertices of $M$ cannot be distinguished by any vertex in $V(G)\setminus M$, that is $N(u)\setminus M = N(v)\setminus M$ for each $u,v\in M$. For example, $\emptyset$, $V(G)$ and all the singletons $\{v\}$ for each $v\in V(G)$ are modules,  they are called \emph{trivial modules}. A graph is \emph{prime} if all its modules are trivial.

Given two modules of $G$, either they are disjoint or one is included in the other. This property leads to a recursive decomposition of a graph, called \emph{modular decomposition}, that refers to the process whereby an entire graph is decomposed; at any stage of the process, the current subgraph being decomposed will be a module of the original graph. Each of these subgraphs is decomposed recursively. This process continues until all the subgraphs being decomposed contain only a single vertex. The modular decomposition of a graph can be computed in linear time  (there are various algorithms, e.g. see~\cite{habib-2010}).

\begin{proposition}\label{prop:module}
If $M$ is a $($largest$)$ module of $G$, then $\Adim(G) \ge |M|$. In particular, if $\Adim(G) = 1$, then $G$ is prime.
\end{proposition}

\begin{proof}
By definition of module, $V(G)\setminus M$ is an $|M|$-ARS of $G$, hence $\Adim(G) \ge |M|$. If follows that if $G$ contains a non-trivial module $M\subseteq V(G)$, then $\Adim(G) \ge |M| > 1$.
\end{proof}

As already observed, given any graph $G$, it is possible to compute its modular decomposition in linear time. Concerning the result, if we get some non-trivial module, then $\Adim(G) \ge |M|$, where $M$ is any largest non-trivial module. Conversely, if $G$ is prime, we have no information about $\Adim(G)$. For instance, the path graph $P_4$ is prime and $\Adim(P_4)=1$, whereas $P_5$ is prime as well, but $\Adim(P_5)=2$.

\subsection{Diameter two graphs}

As we next show, for a graph to be $1$-metric antidimensional it is necessary to have a ``large'' diameter. 

\begin{theorem}
\label{thm:diam-2}
If $\diam(G) = 2$, then $\Adim(G) \ge 2$.
\end{theorem}

\begin{proof}
Assume first that $\delta(G) = 1$ and let $x\in V(G)$ be a vertex with $\deg_G(x) = 1$. Let $y$ be the unique neighbor of $x$. Since $\diam(G) = 2$, each additional vertex of $G$ is adjacent to $y$. Therefore, $\{y\}$ is an $(n(G)-1)$-ARS. Using again the assumption $\diam(G) = 2$ we have $n(G)-1 \ge 2$, so that $\Adim(G)\ge 2$. In the rest, we may thus assume that $\delta(G)\ge 2$.

Let $x\in V(G)$ be an arbitrary vertex. Then by the above assumption, $|L_1(x)|\ge 2$. If also $|L_2(x)|\ge 2$, then $\{x\}$ is an ARS yielding the required conclusion. Hence assume that $|L_2(x)| = 1$, and let $y$ be the unique vertex of $L_2(x)$. If $x$ and $y$ have the same neighbors, then $L_1(x)$ is a $2$-ARS. Hence, assume further that there exists a vertex $z\in L_1(x)$ such that $yz\notin E(G)$. Since $\delta(G)\ge 2$, we have  $|L_1(y)|\ge 2$. In addition, as neither $x$ nor $z$ is adjacent to $y$ we also have $|L_2(y)|\ge 2$. We can conclude that in this subcase $\{y\}$ is a $2$-ARS and we are done.
\end{proof}

\section{Graphs $G$ with $\Adim(G)=1$}\label{sec:adim-1}


In this section, we provide some results concerning 1-metric antidimensional graphs in the context of some well-known graph classes. Before doing this, we start by assessing the complexity of recognizing whether a given graph $G$ is 1-metric antidimensional.

\subsection{Recalling the case of trees}

The first studies on the existence of $1$-metric antidimensional graphs were presented in \cite{trujillo-2016}, where the class of $1$-metric antidimensional trees and unicyclic graphs were characterized and polynomial algorithms to decide if a given tree or a unicyclic graph is such were designed. To describe the case of trees, we need the following terminology and notations from~\cite{trujillo-2016}.

If $T$ is a tree and $u\in V(T)$, by $T_v$ we denote the tree $T$ rooted at $v$. Accordingly, for each vertex $v\in V(T)$, we write $p_u(v)$, $C_u(v)$, and $D_u(v)$ to denote the parent, the children, and the descendants of $v$ in $T_u$, respectively. We assume $p_u(u)=\emptyset$ for the sake of consistency.

\begin{definition}[$u$-branches and $\epsilon$-equivalence in trees]
Let $T$ be a tree and $u\in V(T)$. Given a neighbor $v$ of $u$, a $u$-\emph{branch} of $T_u$ at $v$ is the subtree $T_{u,v}$ induced by $\{u,v\}\cup D_u(v)$. We say that two $u$-branches $T_{u,v'}$ and $T_{u,v''}$ are $\epsilon$-\emph{equivalent} if $\epsilon_{T_{u,v'}}(u)=\epsilon_{T_{u,v''}}(u)$.
\end{definition}

Observe that two $\epsilon$-equivalent branches $T_{u,v'}$ and $T_{u,v''}$ satisfy that for every vertex in $T_{u,v'}$ there exists another in $T_{u,v''}$ with the same distance to $u$. This directly relates to the $k$-antiresolving set concept and leads to the definition of \emph{balancing factor} below.

\begin{definition}[Balancing factor]
Given a vertex $u$ of a tree $T$, we define the \emph{balancing factor} $\xi_T(u)$ as the maximum cardinality of a set of $\epsilon$-equivalent $u$-branches in $T_u$. \end{definition}

Note that if no two $u$-branches in $T_u$ are $\epsilon$-equivalent, then, by definition, $\xi_T(u) = 1$. With the concepts above in mind, the following results were shown in \cite{trujillo-2016}.


\begin{lemma}\label{lem_connected}
\emph{\cite{trujillo-2016}}
Any $k$-antiresolving set $S$ in a tree $T$ with $k\ge 2$ induces a connected
graph.
\end{lemma}

\begin{theorem}\label{theo_one_anti_tree}
\emph{\cite{trujillo-2016}}
A tree $T$ is $1$-metric antidimensional if and only if $\xi_T(v)=1$ for every $v\in V(T)$.
\end{theorem}

The most interesting fact of the characterization above is that it can be translated into a polynomial algorithm that checks whether a given tree is $1$-metric antidimensional. A natural generalization of this study was the one in which an extra edge is added to a tree, given the step to consider unicyclic graphs. For those, the problem can still be solved in polynomial time, which was also made in \cite{trujillo-2016}, together with a structural characterization of those $1$-metric antidimensinal unicyclic graphs. Such characterization relies on the one from Theorem \ref{theo_one_anti_tree}. It is hence of interest to continue with some other characterizations of $1$-metric antidimensional graphs among other generalizations of trees.

\subsection{Geodetic graphs}\label{sec_block}

In this section, we provide a characterization of the class of geodetic graphs which are $1$-metric antidimensional. Recall that a graph is \emph{geodetic} if each pair of distinct vertices is connected by a unique shortest path. Notice that every tree, every complete graph, and every odd-length cycle is geodetic. Moreover, if every biconnected component of a graph is geodetic, then the graph itself is geodetic. In particular, every block graph (a graph in which the biconnected components are complete) is geodetic. Similarly, because a cycle graph is geodetic when it has odd length, every cactus graph in which the cycles have odd length is also geodetic. These cactus graphs are exactly the connected graphs in which all cycles have odd length. More strongly, a planar graph is geodetic if and only if all of its biconnected components are either odd-length cycles or geodetic subdivisions of a four-vertex clique~\cite{Stemple-1968}.


Let $G$ be a geodetic graph, and let $u\in V(G)$. Given any $v\in V(G)$, $v\neq u$, by definition there exists a \emph{unique shortest $u,v$-path} in $G$. The set of all these shortest $u,v$-paths forms a tree rooted at $u$ denoted as $T_u(G)$. The following statement provides the announced characterization of the geodetic graphs that are $1$-metric antidimensional. Note that Theorem~\ref{theo_one_anti_tree} is its direct consequence.

\begin{theorem}\label{theo_one_anti_block}
If $G$ is a geodetic graph, then $\Adim(G) = 1$ if and only if $\Adim(T_u(G)) = 1$ for every $u\in V(G)$.
\end{theorem}

\begin{proof}
($\Leftarrow$) By contradiction, assume $G$ contains a $k$-antiresolving set $S$, where $k > 1$.

Consider first the case $S = \{v\}$, for some $v\in V(G)$.
For the sake of simplicity, denote $T_v(G)$ as $T$. Since $\Adim(T) = 1$, by Theorem~\ref{theo_one_anti_tree} we get $\xi_T(v)=1$. But this implies there exists a vertex $x$ (which is a leaf in $T$) such that $d_T(v,x)>d_T(v,y)$ for every $y \in V(T)\setminus \{v,x\}$. Thus, there does not exist a vertex having the same metric representation as
$x$ with respect to $S$, which is a contradiction.

Now, we consider the case $|S|\ge 2$. Let $w$ be a vertex in $S$ having at least
one neighbor not in $S$.
Let $w_1,\ldots,w_r$ be those vertices in $N(w)\setminus S$ and again, for the sake of simplicity, denote $T_w(G)$ as $T$.
Since $\Adim(T) = 1$, by Theorem~\ref{theo_one_anti_tree} we get $\xi_T(w)=1$. Then, there exists a leaf vertex $z \not \in S$
belonging to some $w$-branch $T_{w,w_i}$ with $i \in [r]$ such that
$d_T(z,w)>d_T(z',w)$ for every leaf $z' \neq z$ belonging to any $w_j$-branch
$T_{w,w_j}$ with $j\in [r]$. Note that such a leaf vertex $z \not\in S$ exists
because $S$ is connected as stated by
Lemma~\ref{lem_connected}. It follows that any vertex $y$ whose
metric representation with respect to $S$ is equal to that of $z$ does not
belong to
a $w$-branch $T_{w,w_j}$ with $j\in [r]$. Consequently, it results
that $y$ belongs to a $w$-branch $T_{w, w'}$ where $w'\in S$. Given that both
$y$ and $z$ have the same metric representation with respect to $S$, we obtain the following:
\begin{align*}
& d_T(z,w')=d_T(y,w'),\\
& d_T(z,w)-1 = d_T(y,w)+1,\\
& d_T(z,w) \neq d_T(y,w).
\end{align*}
It turns out that $y$ and $z$ do not have the same metric representation with
respect to $S$, a contradiction. As a consequence, it
follows that
$\Adim(T) = 1$.

($\Rightarrow$) Assume $\Adim(G) = 1$ and, by contradiction, $T_u(G)$ contains a $k$-antiresolving set for some $u\in V(G)$ and some $k > 1$. For the sake of simplicity, denote $T_u(G)$ as $T$. According to Theorem~\ref{theo_one_anti_tree}, there exists a vertex $v$ such that $\xi_T(v)\ge 2$. Then, the set
$V(T)\setminus (V(T_{v,v'})\cup V(T_{v,v''}))$, where $T_{v,v'}$ and $T_{v,v''}$ are
two $\xi$-equivalent $v$-branches, is a $k$-antiresolving set for some
$k\ge 2$, which is a contradiction.
\end{proof}

%

In~\cite{trujillo-2016}, an $O(n^2)$ algorithm to decide whether a tree is $1$-metric antidimensional is provided. By Theorem~\ref{theo_one_anti_block}, the same algorithm can be used to determine in $O(n^3)$-time whether a geodetic graph $G$ is $1$-metric antidimensional.

We may recall that examples of unicyclic geodetic graphs that are $1$-metric antidimensional are already known from \cite{trujillo-2016}, in which the length of the unique cycles of such graphs is odd and of arbitrarily large order. In addition, all those geodetic graphs with nine and ten vertices are drawn in Figure~\ref{fig:6-geodetic} and~\ref{fig:15-geodetic}, respectively. Notice that some of them are neither trees nor unicyclic graphs.




\begin{figure}[ht]
    \centering
\includegraphics[width=0.65\linewidth]{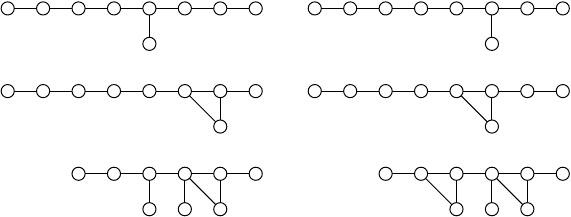}
    \caption{\small The six geodetic graphs $G$ of order $9$ having $\Adim(G)=1$.}
    \label{fig:6-geodetic}
\end{figure}

\begin{figure}[ht]
    \centering
\includegraphics[width=0.65\linewidth]{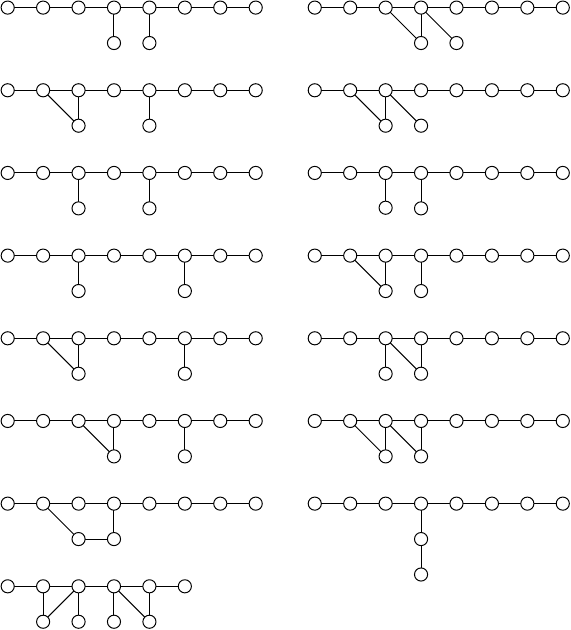}
    \caption{\small The fifteen geodetic graphs $G$ of order $10$ having $\Adim(G)=1$.}
    \label{fig:15-geodetic}
\end{figure}

\subsection{Block graphs}

Theorem~\ref{theo_one_anti_block} provides a characterization of all geodetic graphs $G$ such that $\Adim(G) = 1$. Unfortunately, this characterization does not provide ``structural'' properties of such graphs. Here we consider a subclass of geodetic graphs, namely block graphs, and for these, we provide some structural properties that follow from the results given above. For this sake, recall that two vertices $u,v$ are called \textit{twins} if any vertex $x\ne u,v$ is either adjacent to both of them or to none of them.

\begin{lemma}
\label{lem:blocks}
If $G$ is a block graph such that $\Adim(G)=1$, then the following properties hold.
\begin{enumerate}[{\em (i)}]
\item $\diam(G)$ is odd.
\item Each module of $G$ is trivial (in particular, $G$ does not contain twins).
\item Each block of $G$ contains at most one not cut-vertex.
\item Each pendant block of $G$ is isomorphic to $K_2$.
\end{enumerate}
\end{lemma}
\begin{proof}
(i) If $\diam(G)$ is even, the center of $G$ consists of just one vertex $v\in V(G)$. 
As a consequence, the balancing factor $\xi(v)$ in the tree $T_v(G)$ is at least two. According to Theorems~\ref{theo_one_anti_tree} and~\ref{theo_one_anti_block}, we get $\Adim(G)\ge 2$.

(ii) If $G$ contains a non-trivial module, then $G$ is not prime and hence Proposition~\ref{prop:module} implies $\Adim(G)\ge 2$. $G$ cannot contain twins because a set of twins forms a non-trivial module.

(iii) Assume that a block of $G$ contains two or more not cut-vertices. These vertices are twins of $G$, and hence $\Adim(G)\ge 2$.

(iv) Consider now a pendant block of $G$. This block must be a $K_2$ graph, for otherwise, it contains two or more not cut-vertices.
\end{proof}

Notice that the opposite of Lemma \ref{lem:blocks} is not true in general. That is, there are block graphs $G$ satisfying all four items of the lemma, but $\Adim(G)>1$. To see this, we use the following result, where we also show an infinite family of block graphs that are $1$-metric antidimensional. Let $B_t$ be the block graph obtained from a path $P_{2t}=v_1v_2\dots v_{2t}$ ($t\ge 2$) by adding other $t-1$ isolated vertices $u_1,\dots,u_{t-1}$ and the edges $u_iv_{2i}, u_iv_{2i+1}$ for every $i\in [t-1]$.

\begin{remark}
Let $t\ge 2$ be an integer. Then $\Adim(B_t)=1$ if and only if $t \equiv 1\pmod 2$.
\end{remark}

\begin{proof}
Assume first $t\equiv 0\pmod 2$. Consider the vertex $x=u_{t/2}$. According to the construction of the graph $B_t$, such a vertex is a central vertex of $B_t$, and for every $i\in [\epsilon(x)]$, there are at least two vertices at distance $i$ from $x$. Thus, the set $\{x\}$ is indeed a $2$-ARS of $B_t$, and so, $\Adim(B_t)\ge 2$.

On the other hand, assume $t \equiv 1\pmod 2$. To see that $\Adim(B_t)=1$, we shall run Algorithm ADIM-1 on $V(B_t)$. Since $t \equiv 1\pmod 2$, the center of $B_t$ is not unique, and moreover, for each vertex of $z\in V(B_t)$, there is a unique vertex of degree one (either $v_1$ or $v_{2t}$, say $v_1$ unless $z=v_1$, in which case the vertex in question is taken as $v_{2t}$) such that $d_{B_t}(z,v_1)=\epsilon(z)$. Thus, if $z$ is the first vertex that Algorithm ADIM-1 selects, then in the first step of the algorithm, the vertex $v_1$ (resp. $v_{2t}$ if $z=v_1$) is added to the set $S$ of Algorithm ADIM-1. Due to the structure of the graph $B_t$, in the next step of the algorithm, all the vertices of $V(B_t)$ will be added to the set $S$. Thus, the algorithm will end, and it will return $\Adim(B_t)=1$.
\end{proof}

As a consequence of the remark above, we observe that those block graphs $B_t$, with $t\equiv 0\pmod 2$, satisfy the conditions of Lemma \ref{lem:blocks}, but they are not $1$-metric antidimensional.

\section{Increasing the privacy properties of graphs}\label{sec-increasing-k}

Since a $1$-metric antidimensional graph might not be used to publish sensible information as a social network, it is desirable to perform some changes in such a graph so that its privacy properties increase. Some investigations in this regard were made in \cite{Mauw-2016}. There was first described a privacy-preserving anonymization approach that resists active attacks when the privacy measure $(k, \ell)$-anonymity is considered. Such an anonymization method was based only on some edge addition operations, and preserving the original number of vertices in the social network in question.

In this section, we are also interested in how to embed a given graph $G$ such that $\Adim(G) = 1$ into a larger graph that can hide the property of being 1-metric antidimensional.
To this end, we consider some product graphs in order to ensure the anonymity of the vertices of $G$.

\subsection{Strong product}

The distance function of the strong product is as follows. If $(g, h)$ and $(g', h')$ are vertices of a strong product $G\boxtimes H$, then
\begin{equation}
\label{eq:distance-in-strong}
d_{G\boxtimes H}((g, h), (g', h')) = \max\{ d_G(g, g'), d_H(h, h')\}\,.
\end{equation}
For the proof of~\eqref{eq:distance-in-strong} see~\cite[Proposition~5.4]{hik-2011}.

The following result ensures that the strong product is more friendly for our purposes than the Cartesian product (see Subsection~\ref{sec-Cartesian}), as it allows us to simultaneously embed two graphs $G$ and $H$, for which we have $\Adim(G) = \Adim(H) = 1$ such that (some) anonymity for both $V(G)$ and $V(H)$ is possible.

\begin{theorem}
\label{thm:strong}
If $G$ and $H$ are connected graphs of order at least $2$, then $\Adim(G\boxtimes H) \ge 2$. If in addition each of $G$ and $H$ is of order at least $3$, then $\Adim(G\boxtimes H) \ge 3$. Moreover, in each of the cases, there exists a corresponding $k$-ARS of cardinality $1$.
\end{theorem}

\begin{proof}
By the commutativity of $G\boxtimes H$ we may, without loss of generality, assume throughout the proof that $\diam(G)\ge \diam(H)$.

Let $(g,h)$ be an arbitrary vertex of $G\boxtimes H$.  Consider an arbitrary vertex $g'\in V(G)$ with $d_{G}(g,g') = s\ge 1$. Let $h'$ be an arbitrary neighbor of $h$ in $H$; it exists since $n(H)\ge 2$. By~\eqref{eq:distance-in-strong},
$$d_{G\boxtimes H}((g,h), (g',h)) = d_{G\boxtimes H}((g,h), (g',h'))\,,$$
which in turn implies that $|L_s((g,h))|\ge 2$. As $\diam(G)\ge \diam(H)$ this implies (having~\eqref{eq:distance-in-strong} in mind) that $\{(g,h)\}$ forms a $2$-ARS of $G\boxtimes H$.

Assume now that $n(G)\ge 3$ and $n(H)\ge 3$. In the first subcase, let $H=K_n$, $n\ge 3$. Let $h\in V(H)$. Then, if $g'$ is an arbitrary vertex of $G$ different from $g$, then in view of~\eqref{eq:distance-in-strong} we have $d_{G\boxtimes H}((g,h),(g',h)) = d_{G\boxtimes H}((g,h),(g',h'))$ for every $h'\ne h$. Since $n(H)\ge 3$ this implies that $\{(g,h)\}$ forms a $3$-ARS of $G\boxtimes H$.

In the second subcase, assume that $H$ is not a complete graph. Then $H$ contains three vertices $h, h', h''$, such that $hh'\in E(H)$, $hh''\in E(H)$, and $h'h''\notin E(H)$. Consider now an arbitrary vertex $g'$ of $G$ different from $g$. Using~\eqref{eq:distance-in-strong} again, we infer that
$$d_{G\boxtimes H}((g,h), (g',h)) = d_{G\boxtimes H}((g,h), (g',h')) = d_{G\boxtimes H}((g,h), (g',h''))\,.$$
Setting $s = d_{G\boxtimes H}((g,h), (g',h))$ this implies that $|L_s((g,h))|\ge 3$. Therefore, also in this subcase $\{(g,h)\}$ forms a $3$-ARS of $G\boxtimes H$.
\end{proof}

\subsection{Lexicographic product}

The next result asserts that $\Adim(G)$ for a graph $G$ that represents a lexicographic product is in general quite large, which makes them good candidates to be used for obscuring the privacy of social networks against active attacks.

\begin{proposition}
\label{thm:lexico}
If $G$ is a connected graph of order at least $2$ and $M$ is the largest module of $G$ different from $V(G)$, then $\Adim(G\circ H) \ge |M|\cdot n(H)$.
\end{proposition}

\begin{proof}
By our assumption, $1 \le |M| \le n(G)-1$. We
claim that
$$S = V(G\circ H)\setminus (M\times V(H))$$
is an $(|M|\cdot n(H))$-ARS set. Indeed, let $(g,h)$ and $(g',h')$ be two vertices of $M\times V(H)$. Since $M$ is a module of $G$, we infer that $d_{G\circ H}((g,h), (g'',h'')) = d_{G\circ H}((g',h'), (g'',h''))$ holds for every vertex $(g'',h'')\in V(G\circ H)\setminus (M\times V(H))$.
\end{proof}

\subsection{Cartesian product}\label{sec-Cartesian}

To deal with the Cartesian product, let us first call up the two standard results we need below. First, if $(g, h)$ and $(g', h')$ are vertices of a Cartesian product $G\cp H$, then
\begin{equation}
\label{eq:distance-in-Cartesian}
d_{G\cp H}((g, h), (g', h')) =  d_G(g, g') + d_H(h, h')\,,
\end{equation}
see~\cite[Proposition~5.1]{hik-2011}. Second, if $G$ and $H$ are graphs on at least two vertices, then
\begin{equation}
\label{eq:connectivity-Cartesian}   \kappa(G\cp H) = \min\{\kappa(G)n(H),  \kappa(H)n(G), \delta(G) + \delta(H)\}\,.
\end{equation}
Equation~\eqref{eq:connectivity-Cartesian} was first stated/announced in 1978 by Liouville~\cite{liouville-1978}, and proved only 30 years later by \v Spacapan in~\cite{spacapan-2008}. Its proof can also be found in~\cite[Theorem~25.1]{hik-2011}, where the interesting history behind the formula is also explained.

Consider Hamming graphs $K_n\cp K_n$, $n\ge 4$.
From~\eqref{eq:connectivity-Cartesian} we deduce that $\kappa(K_n\cp K_n) = 2n-2$, and as $K_n\cp K_n$ is $(2n-2)$-regular and of diameter $2$, Corollary~\ref{cor:kappa-regular-diam-2} implies that
$$\Adim(K_n\cp K_n) = 2n-2, n\ge 4\,,$$
a result earlier proved in~\cite[Proposition~3.1(iii)]{fernandez-2023}.

\medskip
Assume that $\Adim(H) = 1$. If one wants to hide the low privacy of this graph in a larger one, then we have the following mild sufficient condition on a graph $G$ which ensures that $\Adim(G\cp H)\ge 2$.

\begin{proposition}
\label{prop:Adim-2-in-Cartesian}
If $G$ is a connected graph with $\#\epsilon(G) \ge 2$ and $H$ is a connected graph with $n(H)\ge 2$, then $\Adim(G\cp H) \ge 2$.
\end{proposition}

\begin{proof}
Since $\#\epsilon(G) \ge 2$, there exists a vertex $g\in V(G)$ with $\#\epsilon_G(g) \ge 2$. Let $g'$ and $g''$ be vertices of $G$ with $d_G(g,g') = d_G(g,g'') = \epsilon_G(g)$. Let now $h$ be an arbitrary vertex of $H$. We claim that $\#\epsilon_{G\cp H}((g,h)) \ge 2$.

Let $h'$ be a vertex of $H$ with $d_H(h,h') = \epsilon_H(h)$. Then using~\eqref{eq:distance-in-Cartesian} we infer that
$$d_{G\cp H}((g,h),(g',h')) = d_{G\cp H}((g,h),(g'',h')) = \epsilon_{G\cp H}((g,h))\,.$$
Hence $\#\epsilon_{G\cp H}((g,h)) \ge 2$. Lemma~\ref{lem:k-connected} completes the argument.
\end{proof}

Let $I_G[u,v]$ be the {\em interval} between $u$ and $v$ in $G$, that is, the set of all vertices of $G$ that lie on some shortest $u,v$-path. By definition,
$\{u,v\}\subseteq I_G[u,v]$. The {\em geodetic number} $g(G)$ of $G$ is the cardinality of a smallest set $S\subseteq V(G)$ such that $\cup_{\{u,v\} \in \binom{S}{2}} I_G[u,v] = V(G)$.

\begin{theorem}
\label{thm:Adim-2-in-Cartesian-for-geodetic-number-2}
If $G$ and $H$ are graphs with $g(G) = g(H) = 2$, then $\Adim(G\cp H) \ge 2$.
\end{theorem}

\begin{proof}
Since $g(G) = g(H) = 2$, there exist vertices $g,g'\in V(G)$ such that $I_G[g,g'] = V(G)$, and vertices $h,h'\in V(H)$ such that $I_H[h,h'] = V(H)$. We claim that $S = \{(g,h), (g',h')\}$ is a $2$-antiresolving set of $G\cp H$.

Let $(g'',h'')$ be an arbitrary vertex from $V(G\cp H)\setminus S$. Since $g(G) = 2$, there exists a shortest $(g,g')$-path $P_G$ in $G$ containing $g''$. Similarly, because $g(H) = 2$, there exists a shortest $(h,h')$-path $P_H$ in $H$ containing $h''$.

Assume first that $g''\cap \{g,g'\} = \emptyset$
and $h''\cap \{h,h'\} = \emptyset$. Let $g'''$ be the neighbor of $g''$ on the $g,g''$-subpath of $P_G$ and let $h'''$ be the neighbor of $h''$ on the $h'',h'$-subpath of $P_H$. Then we have
\begin{align*}
d_{G\cp H}((g''',h'''), (g,h)) & = d_G(g''',g) + d_H(h''',h) \\
& = (d_G(g'',g) - 1) + (d_H(h'',h) + 1) \\
& = d_G(g'',g) + d_H(h'',h)  \\
& = d_{G\cp H}((g'',h''), (g,h))\,.
\end{align*}
It follows that $(g''',h''')$ and $(g'',h'')$ belong to the same equivalence class of $\mathcal{Z}_S$.

Assume second that $g'' = g$. Then $h''\ne h$. Let $g'''$ be the neighbor of $g=g''$ on $P_G$, and let $h'''$ be the neighbor of $h''$ on the $h,h''$-subpath of $P_H$. (It is possible that $h''' = h$.) Then a calculation parallel to the above reveals that $(g''',h''')$ and $(g'',h'')$ belong to the same equivalence class of $\mathcal{Z}_S$.

The remaining three cases are $g'' = g'$, $h'' = h'$, and $h'' = h$. In each of these cases, we can proceed as in the above paragraph to conclude that the vertex  $(g'',h'')$ belongs to an equivalence class of $\mathcal{Z}_S$ of cardinality at least two. We can conclude that $S$ is a $2$-antiresolving set of $G\cp H$.
\end{proof}

In view of Lemma~\ref{lem:k-connected} together with the fact that the Cartesian product of any two graphs of order at least $2$ is $2$-connected, and in view of the results above, it could reasonably be assumed that $\Adim(G\cp H) \ge 2$ holds for any connected graphs $G$ and $H$ of order at least two. Surprisingly, we have checked by computer experiments that $\Adim(T^*\cp P_{2n}) = 1$ for every $n\in[75]$, where $T^*$ is the unique $1$-antidimensional tree of order $7$ (as mentioned at the beginning of Section \ref{sec:adim-1}. This example leads to a more general statement that is next proved.

\begin{proposition}
\label{prop:cp-1}
If $n\ge 1$, then $\Adim(T^*\cp P_{2n}) = 1$.
\end{proposition}

\begin{proof}
Let $n\ge 1$ and set $G = T^*\cp P_{2n}$. Let $V(P_{2n}) = [2n]$, let the consecutive vertices on $P_6$ in $T^*$ be $v_1, \ldots, v_6$, and let $w$ be the vertex of $T^*$ attached to the $v_3$.

To prove the assertion of the proposition, we shall run Algorithm ADIM-1 on $V(G)$. Let $x\in V(G)$ be the first vertex selected by Algorithm ADIM-1. In view of~\eqref{eq:distance-in-Cartesian} we infer that $\#\epsilon_G(x)=1$, where the unique vertex realizing the eccentricity distance from $x$ is from the set $\{(v_1,1),(v_6,1),(v_1,2n),(v_6,2n)\}$. Therefore, at the next step of the algorithm, we either have $\{(v_1,1),(v_6,2n)\}\subseteq S$ or $\{(v_6,1),(v_1,2n)\}\subseteq S$. In either case, Algorithm ADIM-1 afterwards adds the vertices $(w,i)$, $i\in [2n]$, to $S$, as well as, the other two vertices from the set $\{(v_1,1),(v_6,1),(v_1,2n),(v_6,2n)\}$, which have not yet been added. From here it follows that, at the end of the algorithm, $S=V(G)$ holds. We can conclude that $\Adim(G) = 1$.
\end{proof}

With Proposition~\ref{prop:cp-1} in mind, we cannot guarantee anonymity properties of Cartesian product graphs in general.

Now, notice that the grid graphs $P_r\cp P_t$ are covered by Theorem \ref{thm:Adim-2-in-Cartesian-for-geodetic-number-2}, which means that $\Adim(P_r\cp P_t)\ge 2$. Moreover, it is known from \cite{Cangalovic-2018} that this equality indeed holds for the case $r,t$ are even. However, when publishing some sensible information in a social network, we must be aware that some subgraphs of a grid graph can be $1$-metric antidimensional, as we next show.

\begin{proposition}
Let $G$ be a subgraph of $P_{2n}\cp P_{2n}$, $n\ge 1$, obtained by removing only one edge $e$. Then $\Adim(G)=1$ if and only if the end vertices of $e$ are of degree at most $3$.
\end{proposition}

\begin{proof}
Let $V(P_{2n})=[2n]$. The case $P_2\cp P_2$ is trivial since it is a cycle $C_4$ and removing any edge yields a path $P_4$, which is $1$-metric antidimensional. Hence, from now on, we consider $n\ge 2$. We may assume (WLOG) that the edge removed is of the form $e=(i,j)(i,j+1)$ where $i\in [2n]$ and $j\in [2n-1]$.

If one of the end vertices of $e$ has degree $4$, then it can be readily observed that $S=\{(1,1),(2n,2n)\}$ is a $2$-ARS of $G$. Thus, $\Adim(G)\ge 2$ in this case.

Hence, by the symmetry of the Cartesian product, it only remains to consider the case when $e=(1,j)(1,j+1)$. To prove that $\Adim(G)=1$, we shall run Algorithm ADIM-1 on $V(G)$. Let $(i',j')$ be the first vertex selected by Algorithm ADIM-1. Note that $\#\epsilon_G((i',j'))=1$, and moreover the unique vertex realizing the eccentricity distance from $(i',j')$ is from the set $\{(1,1),(1,2n),(2n,1),(2n,2n)\}$. Thus, in the first step of the algorithm, such eccentric vertex is added to the set $S$ from the algorithm. We now have two different situations.

\medskip
\noindent
Case 1: $i',j'\in [n]$ or $i',j'\notin [n]$.\\ Hence, the eccentric vertex of $(i',j')$ added to $S$ (in the first step) is either $(1,1)$ or $(2n,2n)$. This, in turn, implies that in the second step the other of these two vertices is also added to $S$. Namely, at this point, it holds that $S=\{(i',j'),(1,1),(2n,2n)\}$. Note it is possible $(i',j')\in \{(1,1),(2n,2n)\}$. Now, in the next step, Algorithm ADIM-1 will add to $S$ the vertex $(1,j+1)$. Notice that it has already happened that $(1,j+1)=(i',j')$, in which case, we would have arrived at a same configuration in the previous step.

Now, if $j+1\ge n+1$, then the unique eccentric vertex of $(1,j+1)$ is $(2n,1)$, which needs to be added to $S$ in the next step. Consequently, $(1,2n)$ will eventually also be added to $S$. After these additions, all the corner vertices $(1,1),(1,2n),(2n,1),(2n,2n)$ are in $S$, which leads the algorithm to add all the remaining vertices of $G$ to $S$ in the next (and final) step.

\medskip
\noindent
Case 2: $i'\in [n]$ and $j'\notin [n]$; or $i'\notin [n]$ and $j'\in [n]$.\\
In this situation, we proceed similarly to Case 1. The main difference is that in the first step of the algorithm, one of the two vertices $(1,2n)$ or $(2n,1)$ is added to $S$, and afterwards, the other one also. Also, a next step will lead to add the vertex $(1,j)$ to $S$, which will further somehow lead to adding the other two corner vertices $(1,1),(2n,2n)$ as well.

\medskip
In both cases, Algorithm ADIM-1 ends with the set $S=V(G)$. Therefore, it will return the answer $\Adim(G)=1$, which is the desired conclusion.
\end{proof}

\section{Concluding remarks}

We conclude our exposition by pointing out a few open questions that might be of interest as a continuation of this research line.

\begin{itemize}
    \item According to the experimental results from Section \ref{sec-exper}, we have noted that usually all the $1$-metric antidimensional graphs have connectivity at most $2$. In this sense, it is true that if $G$ is a graph with $\kappa(G)\ge 3$, then $\Adim(G)\ge 2$?

    The question above can also be stated in the following way: Find $1$-metric antidimensional graphs $G$ with $\kappa(G)\ge 3$.
    \item Again according to the experimental results from Section \ref{sec-exper}, all the  $1$-metric antidimensional graphs $G$ of order $n$ that we have found satisfy that $|E(G)| \le \frac{n(n-1)}{4}$. This suggests the property that dense graphs are not $1$-metric antidimensional. Hence, is it true that if $G$ satisfies that $|E(G)|\ge \frac{n(G)(n(G)-1)}{4}$, then $\Adim(G)\ge 2$?
    \item In Subsection~\ref{sec_block} we have seen that it can be determined in $O(n^3)$-time whether a given geodetic graph $G$ is $1$-metric antidimensional. A structural characterization of $1$-metric antidimensional graphs remains a challenging open problem.
    \item Subsection \ref{sec-Cartesian} shows the existence of Cartesian product graphs that are $1$-metric antidimensional. In this sense, it is worth of considering the properties that two graphs $G$ and $H$ must satisfy so that $G\cp H$ is $1$-metric antidimensional.
    \item Are there any other transformations that can be developed on $1$-metric antidimensional graphs in order to obscure their privacy properties against active attacks?
\end{itemize}

\section*{Acknowledgements}

This investigation has been supported by the Project HERMES funded by the European Union NextGenerationEU/PRTR via INCIBE. S. Cicerone and G. Di Stefano were partially supported by the European Union - NextGenerationEU under the Italian Ministry of University and Research (MUR) National Innovation Ecosystem grant ECS00000041 - VITALITY - CUP J97G22000170005, and by the Italian National Group for Scientific Computation (GNCS-INdAM). S. Klav\v{z}ar acknowledges the financial support from the Slovenian Research and Innovation Agency ARIS (research core funding P1-0297 and projects N1-0285, N1-0355). I. G. Yero has been partially supported by the Spanish Ministry of Science and Innovation through the grant PID2023-146643NB-I00.

\end{document}